\documentclass[letterpaper]{article} % DO NOT CHANGE THIS
\usepackage{aaai2026}  % DO NOT CHANGE THIS
\usepackage{times}  % DO NOT CHANGE THIS
\usepackage{helvet}  % DO NOT CHANGE THIS
\usepackage{courier}  % DO NOT CHANGE THIS
\usepackage[hyphens]{url}  % DO NOT CHANGE THIS
\usepackage{graphicx} % DO NOT CHANGE THIS
\usepackage{natbib}  % DO NOT CHANGE THIS AND DO NOT ADD ANY OPTIONS TO IT
\usepackage{caption} % DO NOT CHANGE THIS AND DO NOT ADD ANY OPTIONS TO IT
\usepackage{algorithm}
\usepackage{algorithmic}

\usepackage{amsfonts} 
\usepackage{amsmath}
\usepackage{amsthm}
\newtheorem{lemma}{Lemma}
\newtheorem{theorem}{Theorem}
\usepackage{booktabs}
\usepackage{multirow}
\usepackage{subcaption}
\usepackage{newfloat}
\usepackage{listings}
\DeclareCaptionStyle{ruled}{labelfont=normalfont,labelsep=colon,strut=off} % DO NOT CHANGE THIS
\floatstyle{ruled}
\newfloat{listing}{tb}{lst}{}
\floatname{listing}{Listing}
\title{DERMARK: A Dynamic, Efficient and Robust Multi-bit Watermark \\ for Large Language Models}

\author {
    Qihao Lin\textsuperscript{\rm 1},
    Chen Tang\textsuperscript{\rm 1},
    Lan Zhang\textsuperscript{\rm 1},
    Junyang Zhang\textsuperscript{\rm 1},
    Xiangyang Li\textsuperscript{\rm 1}
}
\affiliations {
    \textsuperscript{\rm 1}University of Science and Technology of China\\
    lqh031106@mail.ustc.edu.cn,
    chentang1999@mail.ustc.edu.cn,
    zhanglan@ustc.edu.cn,
    zhangjunyang@mail.ustc.edu.cn,
    xiangyangli@ustc.edu.cn
}

\usepackage{bibentry}
\begin{document}

\maketitle

\begin{abstract}
As large language models (LLMs) grow more powerful, concerns over copyright infringement of LLM-generated texts have intensified. LLM watermarking has been proposed to trace unauthorized redistribution or resale of generated content by embedding identifiers within the text.
Existing approaches primarily rely on one-bit watermarking, which only verifies whether a text was generated by a specific LLM. In contrast, multi-bit watermarking encodes richer information, enabling the identification of the specific LLM and user involved in generated or distributed content. 
However, current multi-bit methods directly embed the watermark into the text without considering its watermark capacity, which can result in failures, especially in low-entropy texts.
In this paper, we analyze that the watermark embedding follows a normal distribution. We then derive a formal inequality to optimally segment the text for watermark embedding. Building upon this, we propose DERMARK, a dynamic, efficient, and robust multi-bit watermarking method that divides the text into variable-length segments for each watermark bit during the inference.
Moreover, DERMARK incurs negligible overhead since no additional intermediate matrices are generated and achieves robustness against text editing by minimizing watermark extraction loss. 
Experiments demonstrate that, compared to SOTA, on average, our method reduces the number of tokens required per embedded bit by 25\%, reduces watermark embedding time by 50\%, and maintains high robustness against text modifications and watermark erasure attacks.

% Our code is open-source
%     \link{Code}{https://anonymous.4open.science/r/DERMARK-2F45}
\end{abstract}

\section{Introduction}
% 强调一下生成的整个sequence能够承载的水印量是不可知的
In recent years, large language models (LLMs) such as GPT-4 \cite{achiam2023gpt} have achieved significant advancements, excelling in various tasks such as instruction following. Training an LLM requires substantial hardware resources, vast amounts of training data, and specialized expert knowledge. Consequently, LLMs are considered valuable intellectual property (IP) of their respective owners.
However, the increasing capabilities of these advanced models pose potential risks of copyright infringement, including unauthorized redistribution or resale of purchased LLM services \cite{birch2023model}. Given these concerns, there is an urgent need for mechanisms to trace the distribution of LLM-generated texts, thereby protecting the IP rights of LLM owners from malicious exploitation.

LLM watermarking has emerged as a promising solution \cite{liu2024survey}. By implicitly embedding watermarks into generated text, model owners can trace the downstream distribution of their outputs.
Most existing watermarking methods focus on verifying whether a given text was generated by a particular LLM—known as one-bit watermarking \cite{kirchenbauer2023watermark}. However, these approaches cannot embed user-specific identifiers and thus fail to support fine-grained attribution.
To address this limitation, multi-bit watermarking has been proposed, which embeds richer information—such as binary strings \cite{WangYC0LM0024}—by adding biases into the LLM logits.
This enables precise attribution of potentially infringing content by identifying both the LLM and the user involved in generation.

The primary factor determining the successful embedding of a multi-bit watermark in LLM-generated text is watermark capacity, defined as the maximum number of bits that can be embedded. This capacity is directly influenced by the entropy of the text. Specifically, the LLM predicts that tokens with relatively uniform probabilities—indicating high entropy—have a larger watermark capacity.
However, existing methods \cite{WangYC0LM0024} ignore this property. They typically divide the text into equal-length segments based solely on the watermark length, embedding one bit per segment. This uniform segmentation strategy significantly degrades the semantics of short or low-entropy texts. In particular, when generating structured content such as code, where entropy tends to be low, watermark embedding often fails.
Therefore, a crucial advancement for multi-bit watermarking in LLM-generated text is to dynamically segment the text based on its entropy. This ensures that each bit is embedded into a segment with sufficient capacity.

% % 挑战
Designing such a multi-bit watermarking method poses three key non-trivial challenges:
\textbf{(1) Lack of a principled segmentation method for LLM-generated text.} 
% 问问这个会不会写的太长
While token entropy serves as a rough proxy for estimating the watermark capacity of generated text, it remains unclear how much entropy is sufficient to reliably embed a single watermark bit. This uncertainty complicates the dynamic assignment of text segments for watermark embedding. 
Furthermore, due to the autoregressive nature of LLMs—where each token is generated based on the previously sampled tokens—embedding a watermark into one token inevitably alters the distribution of subsequent tokens. 
As a result, it becomes impossible to accurately estimate the entropy of an entire segment in advance without already committing to specific token-level modifications. 
This entanglement creates a chicken-and-egg problem: watermark embedding decisions require prior knowledge of entropy, but embedding itself changes the entropy landscape. Effectively resolving this challenge requires a robust mechanism for aligning watermark bits with text segments in real-time during generation—a problem that remains complex and unexplored.
(2) \textbf{Fragility of multi-bit watermarks to text editing.} 
Each bit is embedded within a specific text segment. Consequently, even minor post-editing operations—such as insertions or deletions—can disrupt segment integrity and compromise watermark extraction. This vulnerability is particularly problematic in real-world applications where LLM-generated content is often subject to downstream modifications.
(3) \textbf{Susceptibility to watermark erasure attacks.} 
Since LLM-generated text is fully visible to end users, adversaries may attempt to remove the watermark to evade attribution or facilitate unauthorized redistribution. For example, an attacker could analyze distributional patterns in watermarked text over multiple generations to infer the embedding strategy and reverse-engineer the watermark. Alternatively, paraphrasing the text can obscure or eliminate the embedded signal, rendering the watermark undetectable.

In this work, we theoretically demonstrate that watermark embedding—implemented by perturbing the LLM logits—follows a normal distribution conditioned on these logits.
Based on this, we derive an inequality that enables real-time assessment of whether the currently generated token sequence is sufficient to embed one watermark bit during token generation. This condition enables dynamic estimation of the required segment length for each bit.
Building on this formulation, we propose DERMARK, a dynamic and efficient multi-bit watermarking method. During the embedding phase, DERMARK adaptively segments the text, using the LLM logits to guide the segmentation for each bit.
Since segmentation and embedding rely solely on the logits and do not involve any intermediate matrix computations, the additional time and memory overhead is negligible compared to the overall inference cost.
In the extraction phase, we apply dynamic programming to minimize the segmentation loss (inequality violations from misaligned segments) and the color loss (color imbalance in token distributions within each segment), thereby achieving robustness against text editing.
Furthermore, DERMARK focuses on segmentation, while relying on existing one-bit watermarking methods within each segment. As a result, robustness-enhancing techniques developed for one-bit watermarking can be seamlessly integrated into DERMARK, making it resilient to erasure attacks such as paraphrasing.

Our main contributions are threefold:
(1) We formally show that watermark embedding—achieved by perturbing the LLM logits—follows a normal distribution. This enables the derivation of a closed-form expression for estimating the number of tokens required to embed each watermark bit.
(2) We propose DERMARK, a dynamic, efficient, and robust multi-bit watermarking framework.
It performs real-time segmentation and embedding during inference with minimal overhead. 
A dynamic programming-based extraction method ensures robustness to text edits, while compatibility with existing one-bit watermarking methods ensures resilience against erasure attacks.
(3) Extensive experiments demonstrate that DERMARK achieves superior efficiency, requiring on average 2.26 and 3.7 fewer tokens per embedded bit than SOTA on OPT-1.3b and LLaMA2-7b, respectively. Additionally, it maintains high robustness against text insertion, deletion, and watermark erasure, while incurring minimal inference overhead.

\section{Related Work}

Since 2023, extensive efforts have been made to embed watermarks into LLM-generated text, which can be broadly categorized into two types: one-bit watermarking and multi-bit watermarking.
One-bit watermarking aims to determine whether a given text was generated by a specific LLM~\cite{kirchenbauer2023watermark}. Building upon this foundation, several methods have been proposed to improve robustness against editing, paraphrasing, and other transformations~\cite{liu2024adaptive,hucwwzh24,feng2024certified,liuphm024,wan2024watermarking,guo2024context}, enhance generation quality~\cite{Fu_Xiong_Dong_2024}, or achieve task-agnostic applicability~\cite{masrani2025task}. However, these techniques are inherently limited to binary detection and cannot be directly extended to scenarios requiring user-specific information.
To address this limitation, multi-bit watermarking has been proposed to embed binary strings into LLM-generated text, enabling fine-grained attribution across users and models.
Among them, \cite{yooakak24} extends one-bit watermarking to a multi-bit setting by assigning fixed-length segments to each bit.
~\cite{WangYC0LM0024} incorporates an additional watermark loss during inference to guide bit-wise embedding into fixed-length segments. \cite{zhang2024remark} proposes training a model to encode the watermark into the LLM-generated text. 
However, a key limitation shared by these methods is their reliance on fixed-length segmentation, which fails to consider the variable watermark capacity of different texts. Ignoring this factor can result in embedding failures, particularly in low-entropy or highly structured texts.

To address this limitation, we aim to design a dynamic multi-bit watermarking framework that, during the LLM’s inference phase, dynamically estimates the segment length required for each watermark bit and adaptively segments the output text based on these estimations.

% --------------------------------Problem Formulation-----------------------------------------
\section{Theoretical Analysis}
This section formalizes the multi-bit watermarking problem, introduces key notations, and derives the inequalities that need to be satisfied for watermark embedding.

\subsection{Problem Statement}

In multi-bit watermarking, the goal is to embed a binary watermark consisting of multiple bits into the text generated by an LLM during its inference process. As previously discussed, we divide the text into segments, each dedicated to encoding a single bit. Accordingly, the embedding procedure can be conceptually divided into two steps:
\newline
\textbf{Step 1. Segmentation:} Partition the generated text into segments, each segment corresponding to one watermark bit.
\newline
\textbf{Step 2. Embedding:} Embed each bit into its corresponding segment.

Our primary focus in this work is on \textit{Step 1-the segmentation process}, which critically determines whether the watermark embedding will succeed. 

Once the text has been partitioned, we apply the one-bit watermarking technique proposed by \citet{kirchenbauer2023watermark} to embed each bit into its corresponding segment.

\begin{figure}[t]
    \centering
    \includegraphics[width=0.48\textwidth]{ 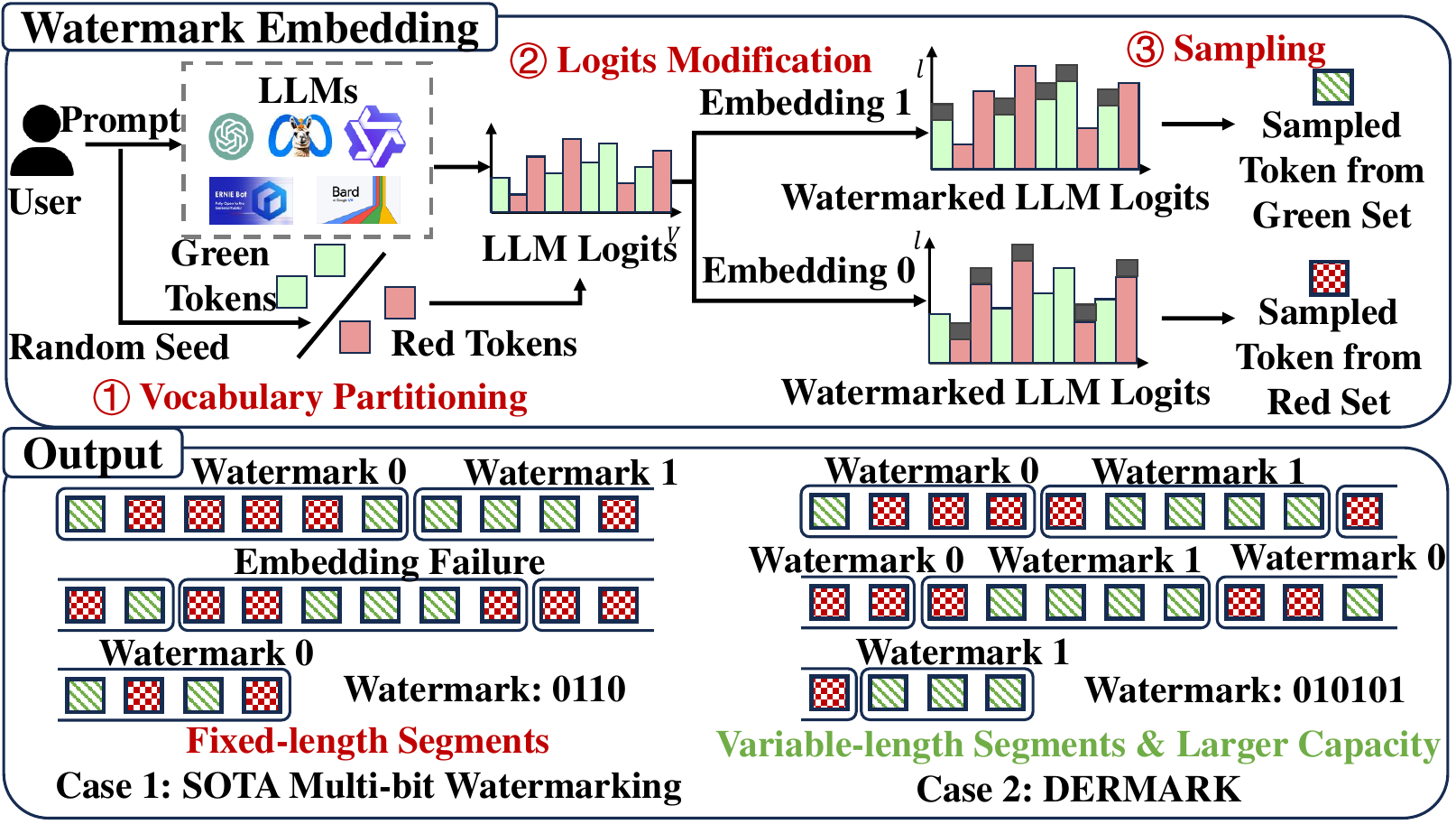}
    \caption{Multi-bit Watermarking Pipeline.}
    \label{fig:process}

\end{figure}

\subsection{Notations and Preliminaries}
To facilitate the understanding of our method, we provide a brief overview of the notation and processes involved in LLM inference, watermark embedding, and extraction.

During the LLM inference phase, let $\mathbf{x}^{p}$ denote the prefix prompt and $\mathbf{s} = \{s^{(0)}, s^{(1)}, \ldots\}$ represent the generated text, where $s^{(t)}$ denotes the $t$-th token. Let $V = \{s_1, \ldots, s_{|V|}\}$ be the vocabulary, where $s_i$ denotes the $i$-th token in the vocabulary.
At each decoding step $t$, the input to the LLM consists of the prompt $\mathbf{x}^{p}$ and the previously generated tokens $\mathbf{s}^{:t-1} = \{s^{(0)}, \ldots, s^{(t-1)}\}$. The LLM produces a vector of logits:
$\mathbf{L}(\mathbf{x}^{p}, \mathbf{s}^{:t-1}) = \{l_1^{(t)}, \ldots, l_{|V|}^{(t)}\}.$
These logits are then passed through the softmax function to obtain the predicted token distribution:
$\mathbf{P}(\mathbf{x}^{p}, \mathbf{s}^{:t-1}) = \{p_1^{(t)}, \ldots, p_{|V|}^{(t)}\}$, where $p_i^{(t)} = e^{l_{i}^{(t)}}/\sum_{j=1}^{|V|} e^{l_j^{(t)}}.$
Finally, the token $s^{(t)}$ is sampled from $\mathbf{P}(\mathbf{x}^{p}, \mathbf{s}^{:t-1})$ using a predefined sampling strategy (e.g., probabilistic sampling).

Multi-bit watermark embedding occurs concurrently with LLM inference. Given a binary watermark $m \in \{0,1\}^K$, each bit $m_k$ is embedded into a segment $S_k$ of consecutive tokens from $\mathbf{s}$. We apply the one-bit watermarking method \cite{kirchenbauer2023watermark} to embed $m_k$ into its corresponding segment $S_k$.
As shown in Fig.\ref{fig:process}, the  watermarking process proceeds as follows:
\newline
\textbf{1. Vocabulary Partitioning.} A random seed is used to partition the vocabulary into a \textit{green list} $G$ and a \textit{red list} $R$ of equal size: $G = \{s_{1}, \ldots, s_{|V|/2}\}, \quad R = \{s_{|V|/2+1}, \ldots, s_{|V|}\}$.
\newline
\textbf{2. Logits Modification}: For bit $m_k$, if $m_k = 1$, a bias $\delta$ is added to the logits of green tokens; if $m_k = 0$, the bias is applied to red tokens: $l_i'^{(t)} = l_i^{(t)} + \delta \cdot \mathbb{I}[s_i \in C_k]$, where $C_k = G$ if $m_k = 1$, and $C_k = R$ if $m_k = 0$.
\newline
\textbf{3. Sampling.} The modified logits $\mathbf{L}'(\mathbf{x}^{p}, \mathbf{s}^{:t-1})$ are passed through the softmax function to obtain the watermarked distribution $\mathbf{P}'(\mathbf{x}^{p}, \mathbf{s}^{:t-1})$, then sample the token $s^{(t)}$.This process is repeated until all tokens in $S_k$ are generated, and the watermark bit $m_k$ is embedded into $S_k$ simultaneously.

In the watermark extraction phase,  
to extract the $k$-th bit $m_k'$ from a given segment $S_k$, we count the number of tokens from $G$ and $R$. If more than half of the tokens in $S_k$ belong to $G$, we decode $m_k' = 1$; otherwise, we set $m_k' = 0$. 
Note that in practical watermark extraction, while individual segments may occasionally yield false positives, the probability that an unwatermarked, human-written text matches a valid multi-bit watermark sequence by chance is negligible.

\subsection{Theoretical Derivation}
Building on the notation and processes introduced above, we now derive the theoretical relationship between watermark bits and the segment length required for embedding.

Since the vocabulary $V$ is partitioned into green and red lists ($G$ and $R$), we begin by analyzing the probability that the next token $s^{(t)}$ belongs to either set prior to watermarking. Let $P^{(t)}_G$ and $P^{(t)}_R$ denote the probabilities that the $t$-th token lies in $G$ or $R$, respectively:
\begin{gather}
    P^{(t)}_G = \frac{\sum_{s_i \in G} e^{l_i^{(t)}}}{\sum_{s_i \in V} e^{l_i^{(t)}}}, \quad
    P^{(t)}_R = \frac{\sum_{s_i \in R} e^{l_i^{(t)}}}{\sum_{s_i \in V} e^{l_i^{(t)}}}.
\end{gather}

Next, we analyze the probabilities after watermarking. Let $P_G'^{(t)}$ and $P_R'^{(t)}$ represent the probabilities of the token belonging to $G$ or $R$ after applying the watermarking bias. We present the following result:
\begin{lemma}
\label{lemma:1}
If $m_k = 1$, 
$P_G'^{(t)} = \frac{e^\delta \cdot P_G^{(t)}}{e^\delta \cdot P_G^{(t)} + (1 - P_G^{(t)})}$; if $m_k = 0$,
$P_R'^{(t)} = \frac{e^\delta \cdot P_R^{(t)}}{e^\delta \cdot P_R^{(t)} + (1 - P_R^{(t)})}$.

\end{lemma}

\begin{proof}
Taking $m_k=1$ as an example.

$P_G'^{(t)} = \frac{e^\delta \sum_{s_i \in G} e^{l_i^{(t)}}}{e^\delta \sum_{s_i \in G} e^{l_j^{(t)}} + \sum_{s_i \in R} e^{l_j^{(t)}}}=\frac{e^\delta \cdot P_G^{(t)}}{e^\delta \cdot P_G^{(t)} + (1 - P_G^{(t)})}.$

\end{proof}

Let $P'^{(t)}$ denote the probability that the next token aligns with the watermarking intent (i.e., sampled from $G$ if $m_k=1$, from $R$ if $m_k=0$). We compute its expected value as:
\begin{equation}\label{eq:Ep}
\mathbb{E}[P'^{(t)}] = P(m_k=1) \cdot P_G'^{(t)} + P(m_k=0) \cdot P_R'^{(t)}.
\end{equation}

\begin{figure*}
 \centering
    \includegraphics[width=\linewidth]{ 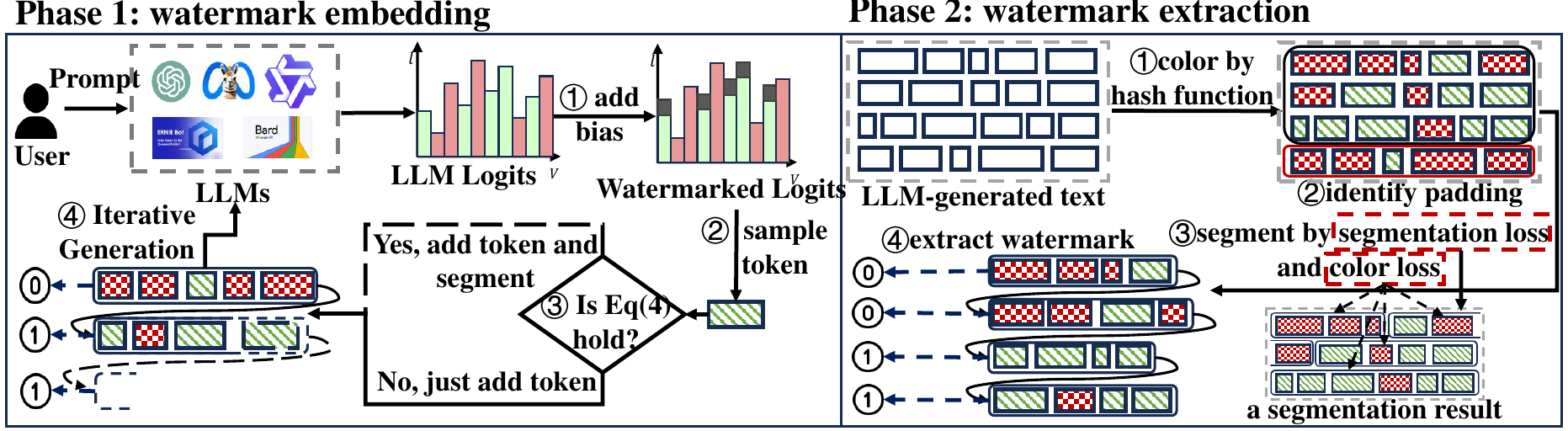}
    \caption{Workflow of DERMARK.}
    \label{fig:workflow}
\end{figure*}

To estimate the fraction of aligned tokens in a segment, let $S$ be a segment of $N$ tokens. Define $X$ as the number of tokens in $S$ that match the current bit embedding requirement, and let $T = X / N$ be the corresponding token proportion. Then, we have the following lemma:
\begin{lemma}
\label{lemma:2}
\begin{gather}
T \sim \mathcal N(\frac{\mu}{N},\frac{\sigma^2}{N^2}), \\
\mu = \sum_{t=1}^N\mathbb E [P'^{(t)}], \sigma^2 = \sum_{t=1}^N (\mathbb{E}[P'^{(t)}] - \mathbb{E}^2[P'^{(t)}]).
\nonumber
\end{gather}
\end{lemma}
\begin{proof}
$\{P'^{(t)}\}_{t=1}^N$ are mutually independent (not identically distributed).
Therefore, $X$ follows a \textbf{Poisson binomial distribution}, i.e., the sum of independent Bernoulli variables.
From the properties of the Poisson binomial distribution, we have:
\[
\mathbb{E}[X] = \sum_{t=1}^N \mathbb{E}[P'^{(t)}],
\mathrm{Var}(X) =\sum_{t=1}^N \left( \mathbb{E}[P'^{(t)}] - \mathbb{E}^2[P'^{(t)}] \right).
\]
By the \textbf{Central Limit Theorem} ~\cite{tang2023poisson}, as $N$ becomes large, the distribution of $X$ can be approximated by a normal distribution:
\[
X \sim \mathcal{N}(\mu, \sigma^2), \quad \text{where } \mu = \mathbb{E}[X], \; \sigma^2 = \mathrm{Var}(X).
\]

Since $T$ is a linear transformation of $X$, it also approximately follows a normal distribution:
\[
T \sim \mathcal{N} \left( \frac{\mu}{N}, \frac{\sigma^2}{N^2} \right).
\]
\end{proof}

To successfully embed one bit, the proportion of aligned tokens $T$ in the segment must exceed $0.5$.
Since $T$ follows a normal distribution, using the significance level $\alpha$, the probability of successful embedding can be estimated as:

\begin{gather}
P\left(T > \frac{1}{2}\right) = \Phi\left( \frac{\mathbb{E}[T] - \frac{1}{2}}{\sqrt{\mathrm{Var}(T)}} \right) \geq 1 - \alpha,
\end{gather}
\begin{gather}
\Rightarrow \Phi^{-1}(1 - \alpha) \leq \frac{\mathbb{E}[T] - \frac{1}{2}}{\sqrt{\mathrm{Var}(T)}}. \label{eq:inequation}
\end{gather}
Here, $\Phi(\cdot)$ is the cumulative distribution function (CDF) of the standard normal distribution, which maps a real-valued input to the probability that a standard normally distributed random variable is less than or equal to that input.

We summarize this result in the following theorem:
\begin{theorem}
\label{thm:1}
When the inequality in Eq.~\eqref{eq:inequation} holds, the watermark bit is embedded into the segment $S$ with confidence at least $1-\alpha$.
\end{theorem}
The significance level $\alpha$ controls the required confidence for watermark embedding. 
Given a fixed $\alpha$, Theorem~\ref{thm:1} shows that Eq.~\eqref{eq:inequation} enables real-time estimation of whether the current segment contains enough tokens to reliably embed the target watermark bit. 
Building on this result, we propose DERMARK—a dynamic, efficient, and robust method for multi-bit watermark embedding—detailed as follows.

\section{DERMARK}

% 随机划分红绿 本质和其他方法是一样的

\subsection{Design Overview}
The overall workflow of our approach is illustrated in Fig.~\ref{fig:workflow}, comprising two main phases: watermark embedding and watermark extraction.
During the embedding phase,  LLM-generated text is dynamically segmented and watermarked based on the current bit and Eq.~\eqref{eq:inequation}.
In the extraction phase, the watermarked text is segmented using dynamic programming, jointly minimizing segmentation loss and color loss to accurately recover the embedded watermark.

\subsection{Multi-bit Watermark Embedding}

During generation, as shown in phase 1 in Fig.~\ref{fig:workflow}, a bias is added to the logits based on the target bit. The modified logits produce a token distribution, from which we sample a token. We then use Eq.~\eqref{eq:inequation} to check whether the current segment can embed the bit. If the inequality holds, the segment ends and the next token starts a new segment. Otherwise, the current segment continues. This process repeats until all bits are embedded.

Nevertheless, to make this process practical for real-world applications, we propose two key enhancements:
\newline
\textbf{(1) Inequality determination.}
To compute Eq.~\eqref{eq:inequation}, we estimate $P_{k1} = P(m_k=1)$ and $P_{k0} = P(m_k=0)$ in Eq.~\eqref{eq:Ep} using the prior distribution of red and green tokens observed within the current segment:
\begin{align}
\label{eq:pk}
P_{k1} = \frac{G_{a:b} + \lambda}{b - a + 2\lambda}, 
P_{k0} = \frac{R_{a:b} + \lambda}{b - a + 2\lambda}.
\end{align}
Here, $a$ and $b$ are the segment's start and current positions; $G_{a:b}$ and $R_{a:b}$ are the counts of green and red tokens. We include a smoothing hyperparameter $\lambda$ for numerical stability.
\newline
\textbf{(2) Redundant tokens.}
If the generated text exceeds the required length for the watermark bits, we use the remaining tokens as padding. We flip the last bit and embed it into these remaining tokens, ensuring these tokens form a distinct final segment. This makes it easy to identify and discard padded content during extraction.

\subsection{Watermark Extraction}

Since the segmentation inequality in Eq.\eqref{eq:inequation} is highly sensitive to token-level perturbations, even minor text edits can lead to different segmentation outcomes, thereby compromising watermark extraction's reliability. To enhance extraction robustness, we propose a dynamic programming-based segmentation strategy tailored for watermark recovery.

Let the segmentation of a sequence $\mathcal{S}$ be denoted as $\text{Seg}(\mathcal{S}) = \{\dots,\mathcal{S}^{(a:b)},\dots\}$. Ideally, during segmentation, the inequality Eq.~\eqref{eq:inequation} is nearly tight, meaning the difference between both sides is sufficiently small.
To quantify deviations from this condition, we define the segmentation loss for each segment as the squared difference between the two sides of Eq.~\eqref{eq:inequation}:
\begin{gather}
    \label{eq:Ls}
    \mathcal{L}_s(a,b) = (f(\mathbb E [P'^{(t)}])-(\Phi^{-1}(1-\alpha))^2 -\epsilon_s)^2, 
    \nonumber
    \\
    f(\mathbb E [P'^{(t)}]) = \frac{(\sum_{t=a}^{b-1}\mathbb E [P'^{(t)}]-\frac{b-a}{2})^2}{ \sum_{t=a}^{b-1}\mathbb{E}[P'^{(t)}]-\sum_{t=a}^{b-1}\mathbb{E}^2[P'^{(t)}]}.
\end{gather}
Here, $a$ and $b$ represent the start and end indices of the segment, respectively, and $\epsilon_s$ is a bias-correction hyperparameter introduced to account for systematic deviations.

Furthermore, due to the presence of the embedded watermark, each segment exhibits a pronounced color imbalance, characterized by a significantly higher proportion of either red or green tokens. To quantitatively capture this effect, we define the color loss for a segment as the normalized difference between the counts of the two token types:
\begin{equation}
    \label{eq:Lc}
    \mathcal{L}_c(a,b) = | min(G_{a:b},R_{a:b})/(b-a) -\epsilon_c |,
\end{equation}
where $G_{a:b}$ and $R_{a:b}$ denote the number of green and red tokens, respectively, within the segment $\mathcal{S}^{(a:b)}$, and $\epsilon_c$ is a bias-correction hyperparameter accounting for systematic deviations in the expected distribution.

As a result, the watermark extraction loss is the sum of the above two losses described above:
\begin{gather}
    \mathcal{L}(\text{Seg}) = \sum (  \mathcal{L}_s(j,i) +\beta  \cdot \mathcal{L}_c(j,i)), \label{eq:total_loss}
\end{gather}
where $\beta$ is a tunable parameter controlling the relative weight of the two components.
To minimize this loss, we employ a dynamic programming approach to identify the optimal segmentation among all possible configurations, with a computational complexity of $\mathcal{O}(N^2)$. The complete watermark extraction workflow is illustrated in Fig.~\ref{fig:workflow} and consists of the following key steps:
\newline
\textbf{1. Token Coloring.} Each token is assigned a color (e.g., red or green) based on a hash of its preceding token.
\newline
\textbf{2. Padding Identification.} Tokens associated with padding are identified by the color distribution. Since padding embeds the inverse of the final bit of the multi-bit watermark, it can be detected based on color imbalance patterns.
\newline
\textbf{3. Segmentation.}
We initialize $\epsilon_s = 0$ and $\epsilon_c = 0$ and define the loss matrix $L[k][p]$, representing the minimum loss incurred by dividing the first $p$ tokens into $k$ segments. A corresponding predecessor matrix $\text{prev}[t][b]$ tracks the start position of the segment ending at position $b$ for $t$ segments. The recurrence relation is defined as: 
\begin{equation*}
   \text{if} \quad L[t-1][a] + \text{cost}[a][b] < L[t][b],
\end{equation*}
\begin{equation*}
    L[t][b] = L[t-1][a] + \text{cost}[a][b], \text{prev}[t][b] = a,
\end{equation*}
where $\text{cost}[a][b]$ denotes the loss of a single segment $\mathcal{S}^{(a:b)}$
% If the total loss for dividing the first $ a $ tokens into $ t-1 $ segments and adding the loss of the segment from $ a $ to $ b $ is smaller than the current loss for dividing the first $ b $ tokens into $ t $ segments, then $ L[t][b] $ is updated to this smaller value. At the same time, the predecessor of $ b $ is recorded as $ a $, preserving the segmentation information.
% By iterating through all possible values of $ b $, $ a $, and $ t $, we ultimately find the minimum value of $ L[k][p] $, which represents the most minor loss for dividing the first $ p $ tokens into $ k $ segments.
\newline
\textbf{4. Bias Parameter Update.} 
Due to the effect of text editing on $\epsilon_s$ and $\epsilon_c$, we iteratively update these parameters based on the current segmentation:
$\epsilon_s'(a,b) = f(\mathbb{E} [P'^{(t)}]) - (\Phi^{-1}(\alpha))^2$ and
$\epsilon_c'(a,b) = \min(G_{a:b},R_{a:b}) / (b-a)$.
The parameters $\epsilon_s$ and $\epsilon_c$ are updated to the min values of $\epsilon_s'(a,b)$ and $\epsilon_c'(a,b)$ across all segments.
\newline
\textbf{5. Iterate Until Convergence.} Steps 3 and 4 are repeated until both $\epsilon_s$ and $\epsilon_c$ converge to stable values.
\newline
\textbf{6. Extract the Watermark from Segments.}
The watermark bit $m_k'$ is extracted from each segment based on the relative proportion of red and green tokens. The overall detection rate is computed as:
\begin{equation}
    d_r = 1-\sum_k \frac{|m_k-m_k'|}{K},
\end{equation}
which quantifies the similarity between the extracted watermark and the original watermark.

\subsection{Performance Analysis}

\textbf{Larger capacity}. Unlike prior methods that allocate a fixed number of tokens per bit, DERMARK introduces a dynamic segmentation strategy that adapts to the token-level watermark capacity, leading to larger capacity. This adaptive segmentation mechanism also improves performance on longer texts, effectively balancing watermark embedding strength and semantic preservation.
\newline
\textbf{Architecture-agnostic}. It only modifies the logits, making it fully compatible with autoregressive LLMs (e.g., GPT, Qwen) without requiring retraining or architecture changes.
\newline
\textbf{Efficient and deployment-friendly}.
Since DERMARK operates only through a lightweight post-processing layer, it avoids intermediate matrix computations, thus achieving watermark embedding and segmentation in linear time with the complexity of $\mathcal{O}(N)$.
Furthermore, the embedding process is compatible with PyTorch’s \texttt{transformers.watermarking\_config} interface, facilitating seamless integration into existing inference pipelines with minimal engineering effort.
\newline
\textbf{Robust against text editing}. Constraining segmentation through a loss-based formulation can effectively reduce the impact of text editing on watermark extraction. The introduction of relaxation terms $\epsilon$ helps to accommodate minor token modifications, thereby reducing sensitivity in evaluating the inequality condition in Eq.~\eqref{eq:inequation}, and maintaining reliable extraction even under perturbations.

% \begin{figure}[t]
%  \centering
%     \includegraphics[width=\linewidth]{ opt_llama_noattack.png}
%     \caption{Capacity Comparison across total dataset on \texttt{OPT-1.3b} (left) and \texttt{LLaMA-2-7b} (right).}
%     \label{fig:noattack}
% \end{figure}

% \begin{figure}[t]
%  \centering
%     \includegraphics[width=\linewidth]{ opt_llama_bottom.png}
%     \caption{Capacity Comparison across low-entropy dataset on \texttt{OPT-1.3b} (left) and \texttt{LLaMA-2-7b} (right).}
%     \label{fig:bottom}
% \end{figure}

\begin{figure*}[t]
    \centering
    \begin{subfigure}[b]{0.48\linewidth}
        \includegraphics[width=\textwidth]{ 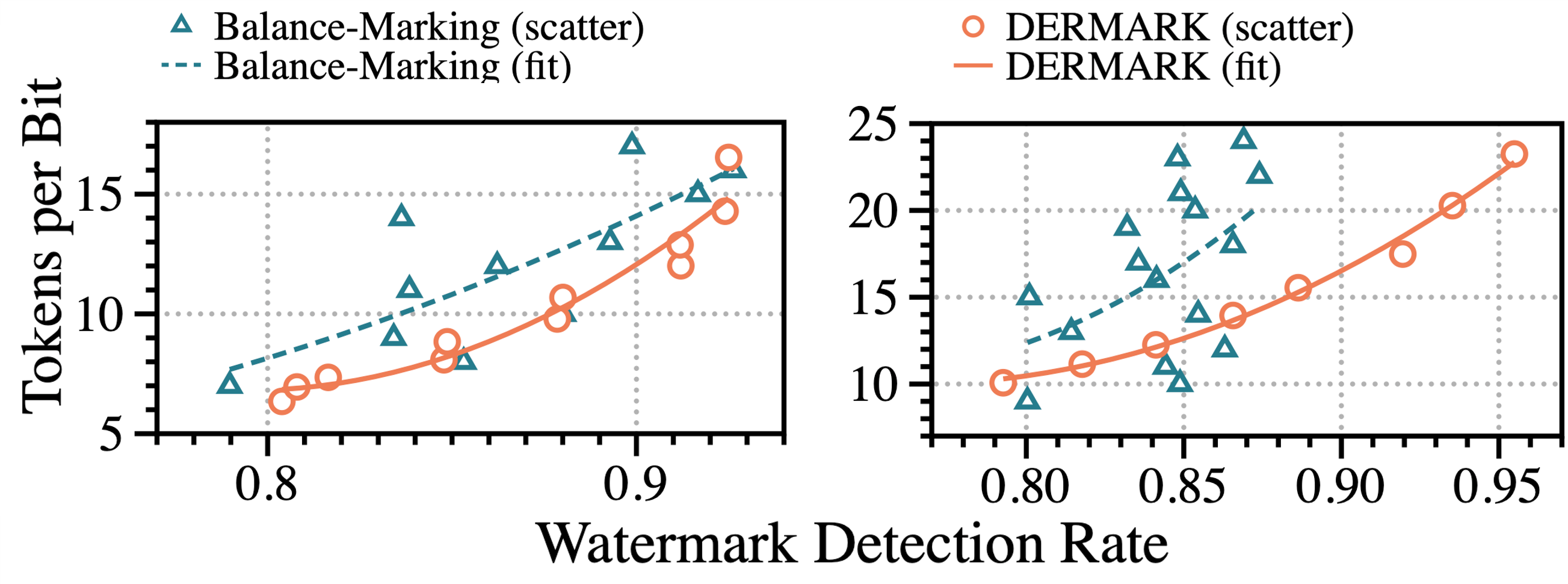}
        \caption{Capacity comparison across total dataset on \texttt{OPT-1.3b} (left) and \texttt{LLaMA-2-7b} (right).}
        \label{fig:noattack}
    \end{subfigure}
    \begin{subfigure}[b]{0.48\linewidth}
        \includegraphics[width=\textwidth]{ 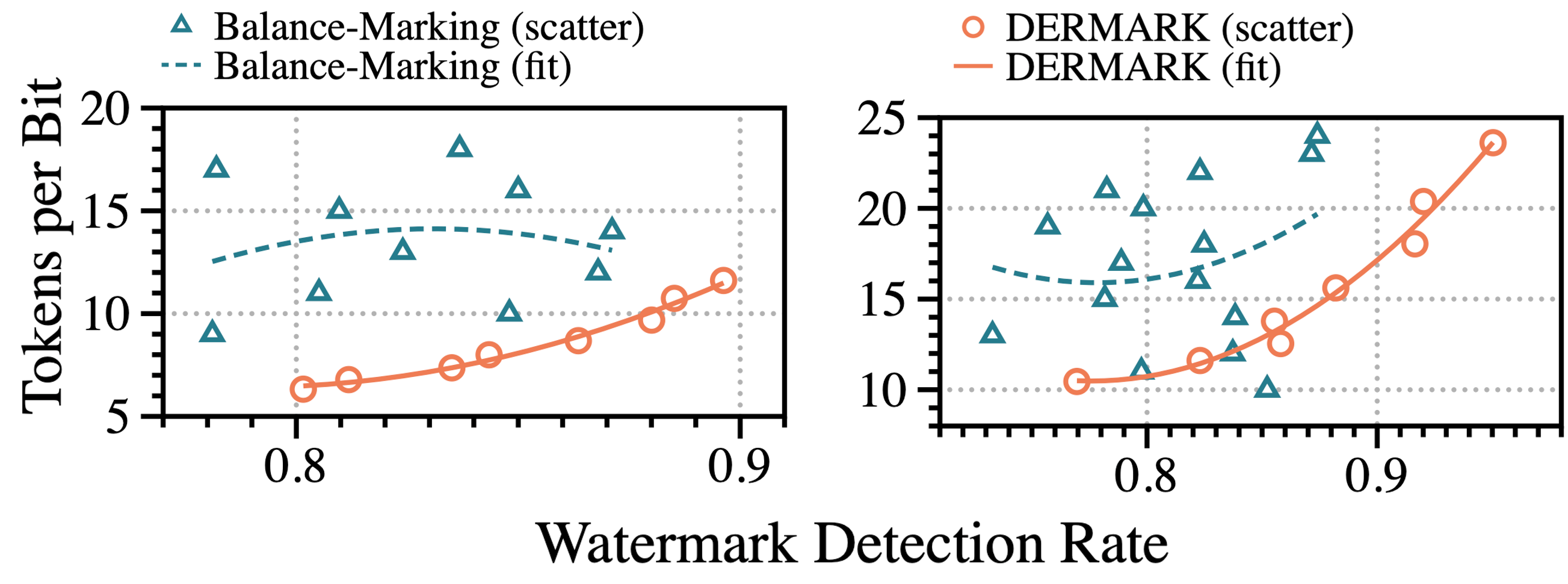}
        \caption{Capacity comparison across low-entropy dataset on \texttt{OPT-1.3b} (left) and \texttt{LLaMA-2-7b} (right).}
        \label{fig:bottom}
    \end{subfigure}
    \caption{Capacity Comparison. Each scatter point in figure is measured over 500 samples under different preset $\alpha$ values.}
    \label{fig:capacity}
\end{figure*}

\section{Experiments}
In this section, we conduct extensive experiments to demonstrate the superior performance of DERMARK in terms of watermark capacity, efficiency, and robustness.

% Each scatter point is measured over 500 samples under different preset $\alpha$ values for DERMARK or Balance-Marking.

\subsection{Experimental Setup}
% baseline 详细说明
All experiments are conducted using two widely adopted language models: \texttt{OPT-1.3b}~\cite{zhang2022opt} and \texttt{LLaMA-2-7b}~\cite{touvron2023llama}. For evaluation, we follow prior work by leveraging the news-like subset of the C4 dataset~\cite{raffel2020exploring} as input prompts and set the watermark strength $\delta$ to 1.
During inference, prompts are randomly sampled from this subset and truncated to the first 100 tokens. Token generation is performed via multinomial sampling, with a repetition penalty of 1.5 to encourage lexical diversity and mitigate mode collapse.
$\alpha$ is set within the range $[0.8, 0.99]$, $\beta$ is fixed at 34, and $\lambda$ is set to $\alpha \cdot (\Phi^{-1}(\alpha))^2$. As a baseline, we adopt Balance-Marking~\cite{WangYC0LM0024}, a state-of-the-art multi-bit watermarking approach, using \texttt{GPT-2} as its auxiliary model.
All experiments were performed multiple times on a server running Ubuntu 24.04.2 LTS and averages were calculated as results. The hardware environment consists of an Intel Xeon Gold 6426Y CPU and four NVIDIA A100-SXM4-80GB GPUs. Python 3.12.8 is used as the primary development environment.

\subsection{Watermark Capacity}
To demonstrate the superior watermark capacity of DERMARK, we first evaluate its performance against Balance-Marking on the evaluation dataset. Watermark capacity is assessed by varying the value of $\alpha$ in DERMARK. Specifically, for each value of $\alpha$, we prompt the LLM to generate 500 texts, from which we extract watermarks to compute the average watermark detection rate and the average number of tokens required to embed each bit. For the baseline, we similarly vary the segment length per bit to obtain comparable measurements.

The results are presented in Fig.~\ref{fig:noattack}. When achieving the same watermark detection rate, DERMARK requires, on average, 2.26 fewer tokens per bit on \texttt{OPT-1.3b}, and 3.7 fewer tokens per bit on \texttt{LLaMA-2-7b}, compared to the baseline. Moreover, on \texttt{LLaMA-2-7b}, the baseline fails to reach a 0.9 watermark detection rate even when the token-per-bit ratio is increased to 25, whereas DERMARK easily surpasses 95\%.

To further highlight DERMARK's capacity, we construct a "low-entropy" dataset comprising the worst-performing 25\% of evaluation samples (based on the watermark detection rate under the baseline). We then evaluate both methods on this challenging subset. As shown in Fig.~\ref{fig:bottom}, DERMARK significantly outperforms the baseline. It requires at least 5 fewer tokens per bit to achieve the same detection rate, yields more stable results with fewer outliers, and continues to improve beyond the baseline's performance plateau—capped at 0.88 on both \texttt{OPT-1.3b} and \texttt{LLaMA-2-7b}.

Finally, we evaluate the watermark capacity under varying watermark strength $\delta$. As shown in Table~\ref{table:traverse_delta}, DERMARK consistently outperforms the Balance-Marking across all $\delta$ values on both \texttt{OPT-1.3b} and \texttt{LLaMA-2-7b}.

% \begin{figure}[t]
%  \centering
%     \includegraphics[width=\linewidth]{ opt_llama_bottom.png}
%     \caption{Capacity Comparison across low-entropy dataset on \texttt{OPT-1.3b} (left) and \texttt{LLaMA-2-7b} (right).}
%     \label{fig:bottom}
% \end{figure}
% \usepackage{booktabs}

% \usepackage{booktabs}
\begin{table}[t]
\centering
\begin{tabular}{@{}c cc cc@{}}
\toprule
\multirow{2}{*}{$\delta$}  & \multicolumn{2}{c}{OPT-1.3b} & \multicolumn{2}{c}{LLaMA-2-7b} \\ \cmidrule(lr){2-3} \cmidrule(lr){4-5} 
      & \#BM      & DERMARK       & \#BM       & DERMARK        \\ \midrule
0.5   & $-$            & 35.12      & $-$             & 45.93          \\
0.8   & 16.00           & 12.59      & $-$             & 22.35          \\
1.0     & 12.80         & 10.78      & 41.89         & 16.25           \\
1.2   & 9.31          & 8.05      & 24.12          & 9.50            \\
1.5   & 7.13         & 5.12       & 17.39          & 7.37       \\
1.8   & 4.57          & 3.48      & 14.06         & 5.33       \\
2.0     & 3.32          & 2.97      & 6.13           & 4.10       \\ \bottomrule
\end{tabular}
\caption{Number of tokens required per watermark bit to achieve a 90\% watermark detection rate under varying watermark strengths $\delta$. \#BM denotes Balance-Marking; “$-$” indicates failure to reach the target accuracy.}
\label{table:traverse_delta}
\end{table}

\begin{figure*}[t]
  \centering
  \begin{subfigure}[b]{0.48\linewidth}
    \includegraphics[width=\linewidth]{ 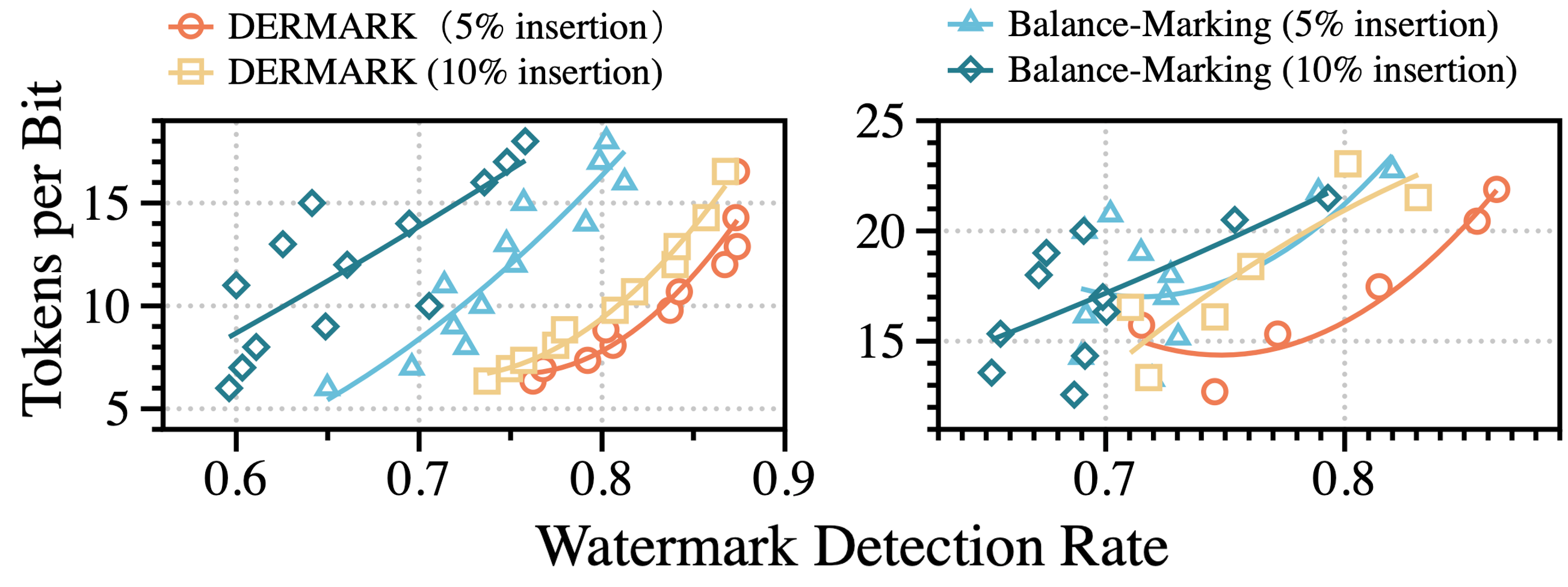}
    \caption{Insert attack comparison on \texttt{OPT-1.3b} (left) and \texttt{LLaMA-2-7b} (right).}
    \label{fig:insert}
  \end{subfigure}
  \hfill
  \begin{subfigure}[b]{0.48\linewidth}
    \includegraphics[width=\linewidth]{ 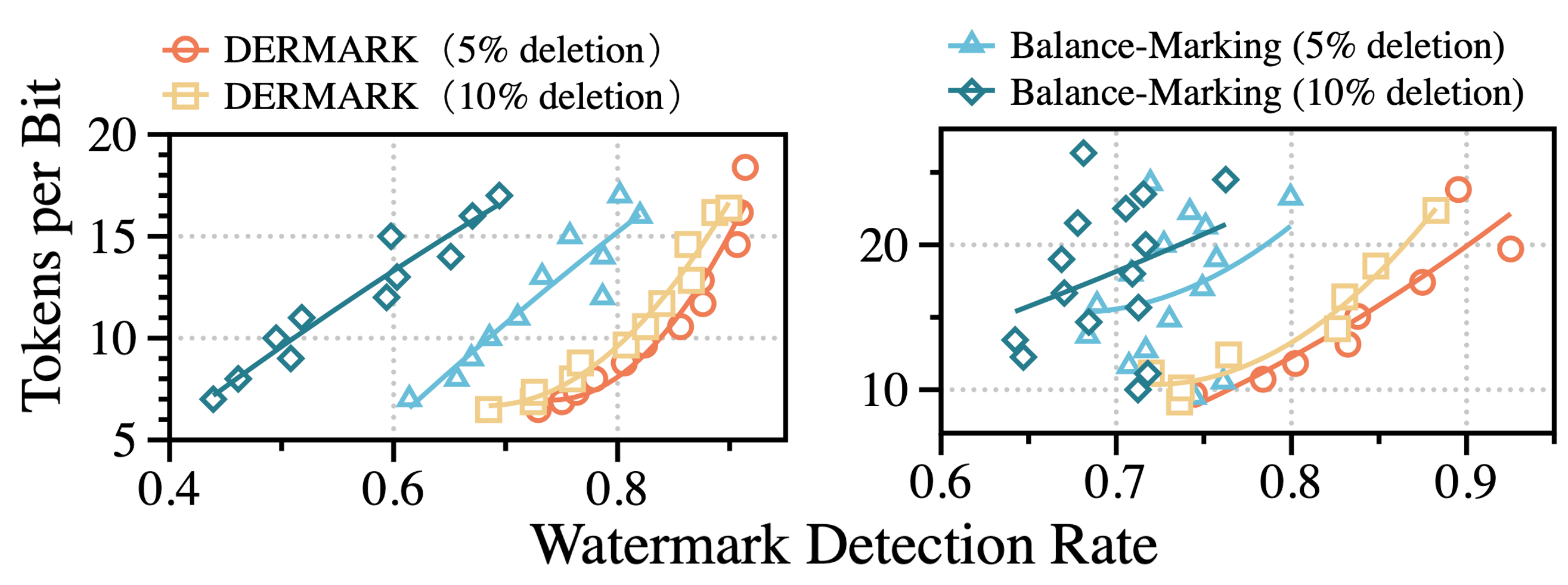}
    \caption{Delete attack comparison on \texttt{OPT-1.3b} (left) and \texttt{LLaMA-2-7b} (right).}
    \label{fig:delete}
  \end{subfigure}
  \caption{Robustness comparison. Each scatter point in figure is measured over 500 samples under different preset $\alpha$ values.}
  \label{fig:insert_delete}
\end{figure*}

\subsection{Impact of Watermarking on Text Quality}
We investigate how watermark strength affects the quality of generated text by varying the value of $\delta$. A total of 500 prompts are randomly sampled from the evaluation set. For each $\delta$ value, we use \texttt{OPT-1.3b} and \texttt{LLaMA-2-7b} to generate responses and evaluate the resulting texts using the Perplexity (PPL) metric computed by \texttt{LLaMA-2-30B}, which serves as a proxy model to assess fluency. Lower PPL values indicate higher textual quality.
Table~\ref{table:combined_ppl} reports the average PPL across all prompts for each watermarking method under different $\delta$ settings. Both methods exhibit similar trends in PPL as $\delta$ varies.
Overall, our approach achieves comparable text quality to the baseline.

\begin{table}[t]
\centering

\begin{tabular}{@{}r cl cl@{}}
\toprule
\multirow{2}{*}{$\delta$} & \multicolumn{2}{c}{OPT-1.3b}                & \multicolumn{2}{c}{LLaMA-2-7b}                 \\ \cmidrule(lr){2-3} \cmidrule(lr){4-5} 
                           & \#BM & DERMARK                  & \#BM & DERMARK                  \\ \midrule
0                           & 215             & 215 \hspace{0.5em} ($\hspace{0.5em}-\hspace{0.5em}$)                  & 132           & 132\hspace{1em}($\hspace{0.5em}-\hspace{0.5em}$)\\
0.5                        & 216             & 213 \hspace{0.5em} ($\downarrow \hspace{0.5em} 3$)      & 142           & 139 \hspace{0.5em} ($\downarrow \hspace{0.5em} 3$)  \\
0.8                        & 228             & 233 \hspace{0.5em} ($\uparrow \hspace{0.5em}  5$)      & 144           & 197 \hspace{0.5em} ($\uparrow 53$)  \\
1.0                        & 232             & 236 \hspace{0.5em} ($\uparrow \hspace{0.5em} 4$)      & 223           & 216 \hspace{0.5em} ($\downarrow \hspace{0.5em} 7$)  \\
1.2                        & 243             & 247 \hspace{0.5em} ($\uparrow \hspace{0.5em} 4$)      & 207           & 162 \hspace{0.5em} ($\downarrow 45$)\\
1.5                        & 224             & 225 \hspace{0.5em} ($\uparrow \hspace{0.5em} 1$)      & 248           & 169 \hspace{0.5em} ($\downarrow 79$)  \\
1.8                        & 253             & 262 \hspace{0.5em} ($\uparrow \hspace{0.5em} 9$)      & 169           & 218 \hspace{0.5em} ($\uparrow 49$)\\
2.0                        & 289             & 256 \hspace{0.5em} ($\downarrow 33$)      & 300           & 267 \hspace{0.5em} ($\downarrow 33$)   \\ \bottomrule
\end{tabular}%
\caption{Perplexity comparison under varying watermark strengths $\delta$. $\delta=0$ denotes unwatermarked model outputs.}
\label{table:combined_ppl}
\end{table}

\begin{figure}[t]
 \centering
    \includegraphics[width=\linewidth]{ 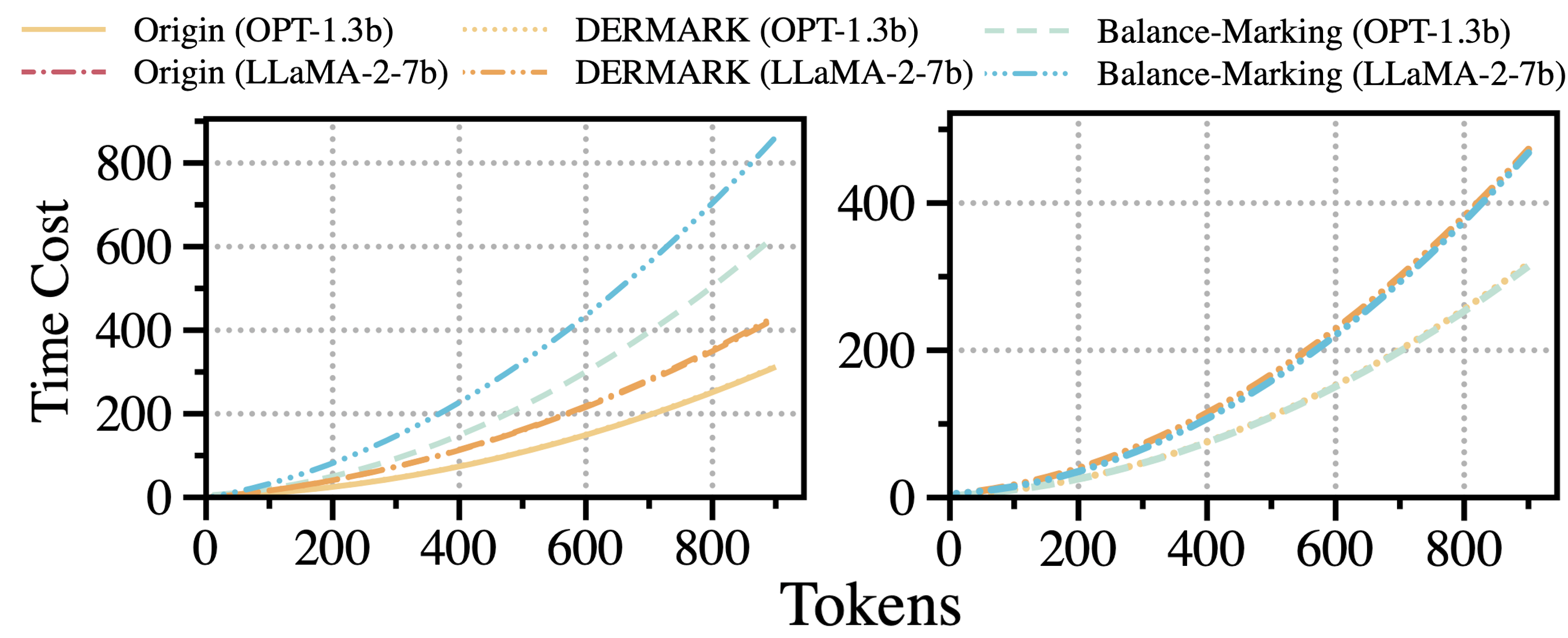}
    \caption{Comparison of time cost in text generation (left) and watermark extraction (right).}
    \label{fig:time_cost}
\end{figure}

% \begin{figure}[t]
%  \centering
%     \includegraphics[width=\linewidth]{ opt_llama_insert.png}
%     % \vspace{-0.25in}
%     \caption{Insert attack comparison on \texttt{OPT-1.3b} (left) and \texttt{LLaMA-2-7b} (right).}
%     \label{fig:insert}
%     % \vspace{-0.05in}
% \end{figure}

% \begin{figure}[t]
%  \centering
%     \includegraphics[width=\linewidth]{ opt_llama_delete.png}
%     % \vspace{-0.25in}
%     \caption{Delete attack Comparison on \texttt{OPT-1.3b} (left) and \texttt{LLaMA-2-7b} (right).}
%     \label{fig:delete}
%     % \vspace{-0.15in}
% \end{figure}

% \begin{figure}[t]
%     \centering
%     \begin{subfigure}[Insect.]
%         \label{fig:insert}
%         \includegraphics[width=0.4\textwidth]{ opt_llama_insert.png}
%         \caption{1}
%     \end{subfigure}
%     \begin{subfigure}[Delete.]
%         \label{fig:delete}
%         \includegraphics[width=0.4\textwidth]{ opt_llama_delete.png}
%         \caption{2}
%     \end{subfigure}
%     \caption{Comparison of time cost in text generation (left) and watermark extraction (right).}
%     \label{fig:robustness}
% \end{figure}

\subsection{Efficiency}
We evaluated the text generation and watermark extraction efficiency of both Balance-Marking and DERMARK. For Balance-Marking, the LLM used for text generation also served as the auxiliary model, with a fixed watermark embedding length of 10 tokens per bit. For DERMARK, we set the error rate to 10\%. Both methods were tested with a fixed prompt length of 100 tokens.

As shown in Fig.\ref{fig:time_cost}, during the watermark embedding phase, DERMARK demonstrates almost identical time overhead compared to the raw outputs of models, which is significantly lower ($<50\%$) than that of Balance-Marking.  
% Moreover, DERMARK obtains better robustness with a watermark extraction time comparable to that of Balance-Marking.
Overall, DERMARK incurs lower time costs and is more suitable for large-scale text watermarking tasks.

\subsection{Robustness}
We evaluate the robustness of DERMARK against three types of attacks: insertion, deletion, and watermark erasure to demonstrate that DERMARK maintains high watermark detection accuracy under such attacks.

\paragraph{Insertion Attacks.}
We assess robustness to insertion by adding 5\% and 10\% random tokens to the generated texts. As shown in Fig.~\ref{fig:insert}, DERMARK consistently achieves higher watermark detection rates with fewer tokens per embedded bit compared to Balance-Marking. The scatter plots and fitted trend lines indicate that DERMARK's performance curve lies below that of the baseline, underscoring its resilience to insertion noise.

\paragraph{Deletion Attacks.}
To evaluate deletion robustness, we randomly remove 5\% and 10\% of tokens from generated outputs. As illustrated in Fig~\ref{fig:delete}, DERMARK outperforms Balance-Marking, achieving a lower token-per-bit ratio while maintaining higher detection accuracy. The performance gap widens at higher thresholds, highlighting DERMARK's superior robustness under adversarial conditions.

\paragraph{Erasure Attacks.}
The core mechanism of DERMARK lies in its segmentation process rather than the embedding step. Within each segment, a one-bit watermarking scheme is employed, allowing DERMARK to leverage existing one-bit watermarking techniques originally designed to enhance robustness against watermark erase attacks.
For example, DERMARK imposes no constraints on the choice of watermarking strategy, ensuring compatibility with both semantically neutral watermarking strategy~\cite{hucwwzh24} and flexible watermarking strategy~\cite{wanggwy25}. This design choice contributes to improved resilience against paraphrasing attacks.

\section{Conclusion}
In this work, we investigate strategies for dynamically assigning text segments to individual watermark bits and propose DERMARK, a lightweight framework that satisfies dynamic embedding requirements with negligible overhead. Extensive experiments on watermark capacity, efficiency, and robustness demonstrate that DERMARK provides a practical and effective solution for multi-bit watermark embedding.


\bibliography{aaai2026}

\newpage

\appendix

\section*{Limitations}
Our method has two primary limitations. First, it is applicable only to autoregressive large language models, with the watermark embedded solely in textual outputs. It is not designed for models generating non-textual content, such as those used in text-to-image generation. Second, the method is not robust against aggressive text modifications, such as copy-paste attacks, heavy perturbations, or more sophisticated text sanitization techniques. These limitations are common to all multi-bit watermarking approaches for large language models—particularly the second one. This vulnerability primarily arises from the dispersion of multiple bits throughout the text, making them more susceptible to bit flips under substantial edits, which in turn hinders watermark extraction.

Nevertheless, we urge reviewers to focus on our core innovation: the dynamic segmentation mechanism. This design effectively leverages the text's capacity to carry watermarks, enabling high-confidence multi-bit embedding without significantly affecting semantic quality. Moreover, in terms of robustness, our method demonstrates substantial improvements over the baseline across a range of adversarial scenarios.

% 讲一下baseline的选择
% 

\section{Baseline Selection}
We acknowledge that the baseline method used in our experiments is not the most recent work~\cite{yoo-etal-2024-advancing} in the field. 
In contrast, we select the SOTA method Balance-Marking as the baseline.
This choice was made intentionally to ensure a clear and fair comparison focused on core methodological differences.

Here we explain why we choose Balance-Marking as the baseline instead of \textit{MPAC}~\cite{yoo-etal-2024-advancing}:

\textbf{(1) Methodological Innovation in Balance-Marking.}
Balance-Marking introduces a principled and effective approach to constructing watermark vocabularies. It selects top-ranked tokens under an auxiliary model until a predefined cumulative probability threshold (e.g., 0.5) is reached, forming the green list.
This design explicitly leverages token likelihoods in a probabilistic manner, thereby improving both the efficiency and robustness of watermark embedding.
In contrast, \textit{MPAC} merely modifies the position allocation strategy of the existing single-bit method \textit{KGW} when extended to the multi-bit setting. It does not propose any novel mechanism for vocabulary selection or multi-bit encoding. As such, it lacks methodological advancement beyond the use of positional cues.

\textbf{(2) Incomplete and Biased Use of Balance-Marking in \textit{MPAC}.}
Although \textit{MPAC} includes Balance-Marking as a baseline in its experiments, it does not faithfully reproduce the original method.
Instead, it employs a simplified variant, citing incompatibility between the tokenizer of the main model and that of the auxiliary model required by Balance-Marking.
However, this issue was explicitly addressed in the original Balance-Marking paper, which proposes a straightforward solution: using the target model itself as the auxiliary model to ensure tokenizer compatibility.
By disregarding this solution, \textit{MPAC} ends up using a weakened version of Balance-Marking, resulting in a baseline comparison that underrepresents its true performance.

% performance比较 能不能有量化结果
\textbf{(3) Misleading Performance Gains via Aggressive Parameter Tuning.}  
In its experiments, \textit{MPAC} sets the watermark strength parameter $\delta$ to 2, which is significantly higher than the widely adopted default value of $\delta = 1$ in most prior watermarking studies.  
It is important to note that in \textit{KGW}—the base method upon which \textit{MPAC} builds—the watermarking strength scales exponentially with $\delta$, i.e., the intensity is proportional to $e^{\delta}$.  
Thus, increasing $\delta$ from 1 to 2 results in approximately a $7.39\times$ increase in watermark strength, artificially boosting detectability.

However, this apparent gain comes at the cost of text quality and practicality, making such parameter settings unsuitable for fair evaluation.  
For instance, under $\delta = 2$, \textit{MPAC} achieves a watermark detection accuracy of 0.899 when embedding 24 bits in 250 tokens, equivalent to 10.42 tokens per bit.  
In contrast, our \textit{Balance-Marking} method reaches 0.900 accuracy at just 6.13 tokens per bit, and our proposed \textit{DERMARK} further improves this to 0.900 accuracy with only 4.10 tokens per bit—demonstrating markedly better capacity-efficiency tradeoffs under standard $\delta = 1$ settings.

Moreover, since \textit{MPAC} only modifies the position allocation mechanism without enhancing the vocabulary construction or bit encoding strategies, its ability to simultaneously improve watermark capacity and generation quality is inherently limited.  
Consequently, the observed performance gains are largely attributable to aggressive hyperparameter tuning rather than substantive methodological improvements.

In summary, we select Balance-Marking as our baseline due to its clear methodological contributions, faithful and complete implementation, and adherence to fair and standardized evaluation protocols. While it is not the most recent work in the field, it remains the most effective among existing multi-bit watermarking methods in terms of robustness and detection performance. This choice ensures a more meaningful and rigorous comparison within the multi-bit watermarking landscape.

\section{Detailed Proofs}
\subsection{Proof of Lemma 1} \label{sec:lamma_1}
\begin{proof}
Taking $m_k=1$ as an example.

For $s_i\in G$, the predicted probability of $s_i$ is:

$p_{Gi}'^{(t)}(s_i \mid\mathbf{x}^{prompt}, \mathbf{s}^{:t-1}) = \frac{e^{l_{i}^{(t)} + \delta}}{\sum_{s_i \in G}e^{l_{j}^{(t)} + \delta} + \sum_{s_i \in R}e^{l_j^{(t)}}}.$

For $s_i\in R$, the predicted probability of $s_i$ is:

$p_{Ri}'^{(t)}(s_i \mid\mathbf{x}^{prompt}, \mathbf{s}^{:t-1}) = \frac{e^{l_i^{(t)}}}{\sum_{s_i \in G}e^{l_{j}^{(t)} + \delta} + \sum_{s_i \in R}e^{l_j^{(t)}}}.$

It is easy to see that $P_G'^{(t)}$ is the sum of $p_{Gi}'^{(t)}$:  

$P_G'^{(t)} = \sum_{s_i \in G} p_{Gi}'^{(t)}(s_i \mid\mathbf{x}^{prompt},  s^{:t-1}) \\
\Rightarrow P_G'^{(t)} = \frac{e^\delta \sum_{s_i \in G} e^{l_i^{(t)}}}{e^\delta \sum_{s_i \in G} e^{l_j^{(t)}} + \sum_{s_i \in R} e^{l_j^{(t)}}}=\frac{e^\delta \cdot P_G^{(t)}}{e^\delta \cdot P_G^{(t)} + (1 - P_G^{(t)})}.$

\end{proof}

\subsection{Proof of Lemma 2} \label{sec:lamma_2}
\begin{proof}
As defined above, $X$ counts the number of aligned tokens in the segment $S$, where each $P'^{(t)} \in (0, 1)$ indicates whether the $t$-th token matches the bit embedding requirement. Since token generation is independent and the target list is re-sampled at each time step, the random variables $\{P'^{(t)}\}_{t=1}^N$ are mutually independent but not identically distributed.

Therefore, $X$ follows a \textbf{Poisson binomial distribution}, i.e., the sum of independent (but non-identical) Bernoulli variables.

From the properties of the Poisson binomial distribution, we have:
\[
\mathbb{E}[X] = \sum_{t=1}^N \mathbb{E}[P'^{(t)}],
\mathrm{Var}(X) =\sum_{t=1}^N \left( \mathbb{E}[P'^{(t)}] - \mathbb{E}^2[P'^{(t)}] \right).
\]
 % \sum_{t=1}^N \mathrm{Var}(P'^{(t)}) = 
By the \textbf{Central Limit Theorem} (CLT) for Poisson binomial distributions (cf.~\cite{tang2023poisson}), as $N$ becomes large, the distribution of $X$ can be approximated by a normal distribution:
\[
X \sim \mathcal{N}(\mu, \sigma^2), \quad \text{where } \mu = \mathbb{E}[X], \; \sigma^2 = \mathrm{Var}(X).
\]

Since $T$ is a linear transformation of $X$, it also approximately follows a normal distribution:
\[
T \sim \mathcal{N} \left( \frac{\mu}{N}, \frac{\sigma^2}{N^2} \right).
\]
\end{proof}

\section{Algorithm}

\subsection{Watermarked Text Generation}

To embed a watermark into a sequence of tokens, we propose a segmentation approach founded on the following inequality:

\begin{gather} 
\Phi^{-1}(1 - \alpha) \leq \frac{\frac{1}{2} - \mathbb{E}(T)}{\sqrt{\text{Var}(T)}}. 
\label{eq:7} 
\end{gather}

This inequality serves as the foundation for identifying segment boundaries during the watermark extraction process. Below, we present Algorithm \ref{alg:1}, which details the watermark extraction procedure.
The algorithm can be divided into several logical groups:
\begin{algorithm}[t]
    \caption{Text Generation with DERMARK}
    \label{alg:1}
    \renewcommand{\algorithmicrequire}{\textbf{Input:}}
    \renewcommand{\algorithmicensure}{\textbf{Output:}}
    
    \begin{algorithmic}[1]
        \REQUIRE Prompt: $\mathbf{x}^{prompt} $, Bit Error Rate: $\alpha$, Binary message: $\mathcal{M}$, Maximum generation length: L%%input
        \ENSURE A token string that carries  $\mathcal{M} $: $\mathcal{S}$  %%output
        \STATE Append the $(1-\mathcal{M}[k-1])$ token to $\mathcal{M}$ 
        \STATE $i \gets 0$
        \STATE $P \gets \{\}$
        \FOR{$t = 0,1,\dots,L$ do}
            \STATE $k \gets \mathcal{M}[i]$  
            \STATE Apply the LLM to prior tokens $\{\mathbf{x}^{prompt}, s^{(0)},\dots, s^{(t-1)}\}$ to get a logit vector $\mathbf{L}(\mathbf{x}^{prompt},\mathbf{s}^{:t})$
            \STATE Using $s^{(t - 1)}$to seed a randomly partition the vocabulary into two identically sized sets: $G$, $R$.
             \STATE Based on bit $k$, choose either $G$ or $R$ to apply gain enhancement.
            \STATE Using soft watermark method, watermark the token  sample the next token, $s^{(t)}$
            \STATE Substitute $G$, $R$, $\alpha$, and $l^{(t)}$ into Eq.\eqref{eq:7} to compute $\mathbb{E}[P'^{(t)}]$.
            \STATE Add  $\mathbb{E}[P'^{(t)}]$ to set $P$
            \IF{\eqref{eq:7}is satisfied on $P$}
               \STATE $P \gets \{\}$
               \IF{$i \leq k$}
                \STATE $i \gets  i+1$
                \ENDIF
            \ENDIF
        \ENDFOR
        \RETURN Outputs
    \end{algorithmic}
\end{algorithm}

    \paragraph{Input and Output:} 
    The input includes a prompt $\mathbf{x}^{prompt}$, a binary message $\mathcal{M}$, a predefined Bit Error Rate (BER) $\alpha$, and a maximum generation length $L$. 
    The output is a generated token sequence $\mathcal{S} = \{s^{(0)}, \ldots, s^{(L)}\}$ that encodes the binary message $\mathcal{M}$ using dynamic segment marking.
    
    \paragraph{Initialization (Lines 1-3):}
    An empty set $P$ is initialized to temporarily store the computed expectations $\mathbb{E}[P'^{(t)}]$. 
    The index counter $i$ is set to $0$, which tracks the current position in the binary message $\mathcal{M}$.
    
    \paragraph{Token Processing (Lines 5-11):}For each step $t$ in the token generation process:
        \begin{itemize}
            \item Retrieve the current bit $k = \mathcal{M}[i]$ (Line 5).
            \item Compute the logit vector $\mathbf{L}(\mathbf{x}^{prompt}, \mathbf{s}^{:t})$ based on the sequence of prior tokens, using the large language model (Line 6).
            \item Partition the vocabulary into two equal subsets $G$ and $R$, using the previous token $s^{(t-1)}$ as a random seed (Line 7).
            \item Apply gain enhancement to $G$ or $R$ depending on the value of $k$ to modulate the probability distribution (Line 8).
            \item Use a soft watermarking method to sample the next token $s^{(t)}$ from the adjusted probability distribution (Line 9).
            \item Compute $\mathbb{E}[P'^{(t)}]$ using $G$, $R$, $\alpha$, and the current logit vector, and append it to the set $P$ (Line 10).
        \end{itemize}
    
    \paragraph{Segmentation and Message Update (Lines 12-17):}
    \begin{itemize}
        \item Periodically check whether the consistency condition in Eq.\eqref{eq:7} is satisfied for the accumulated probabilities in $P$ (Line 13).
        \item If the condition is satisfied:
        \begin{itemize}
            \item Compare the counters $G$ and $R$ to determine the corresponding bit for the segment, and update the binary message $\mathcal{M}'$ (Lines 14-16).
            \item Increment the index counter $i$ to proceed to the next bit in $\mathcal{M}$ (Line 17).
        \end{itemize}
    \item Reset $P$, $G$, and $R$ for processing the next segment (Lines 18-19).
    \end{itemize}

\begin{algorithm}[t]
    \caption{Robust Segmentation}
    \label{alg:2}
    \renewcommand{\algorithmicrequire}{\textbf{Input:}}
    \renewcommand{\algorithmicensure}{\textbf{Output:}}
    \begin{algorithmic}[1]
        \REQUIRE A token string that carries  $\mathcal{M} $: $\mathcal{S} =\{s^{(0)}, \ldots, s^{(L)}\}$, Bit Error Rate: $\alpha$ , bit number:k%%input
        \ENSURE Segmentation for the Token string:Segments%%output
        \STATE Calculate the membership of each token $s^{(t)}$ in either $V_0$ or $V_1$, and store the result in bit
        \STATE Calculate the $P^{(t)}_0$ for each token $s^{(t)}$, and store the result in $P$
        \STATE Initialize two 2D arrays $\mathcal{L}$ and $\text{prev}$ of size $\text{k} \times (\text{L}+1)$
        \FOR{$p = 1,2,\dots ,\text{k}$ do}
            \FOR{$q = 1,2,\dots ,\text{L }+ 1$ do}
                \STATE $\mathcal{L}[p][q] \gets [\infty]* (L+1)$ 
            \ENDFOR
        \ENDFOR
        \FOR{$t = 1,2,\dots ,k$ do}
            \FOR{$b = 1,2,\dots ,L$ do}
                \FOR{$a = 0,1,\dots ,b-1$ do}
                    \STATE  compute $\text{cost}[a][b]$
                    \IF{$\mathcal{L}[t-1][a]+ \text{cost}[a][b] < \mathcal{L}[t][b]$}
                        \STATE $\mathcal{L}[t][b] \gets \mathcal{L}[t-1][a]+\text{cost}[a][b] $
                        \STATE $\text{prev}[t][b] \gets a$
                    \ENDIF
                \ENDFOR
            \ENDFOR
        \ENDFOR
        \STATE $\text{Segments} \gets []$ 
        \STATE $\text{current} \gets n$
        \WHILE{$\text{current} > 0$}
            \STATE $\text{start} \gets \text{Seg}[\text{current}]$
            \STATE Append $(\text{start}, \text{current})$ to Segments
            \STATE $\text{current} \gets \text{start}$
        \ENDWHILE
        \RETURN Segments
    \end{algorithmic}
\end{algorithm}
\subsection{Watermark Extraction}
Algorithm \ref{alg:2} outlines a method for reconstructing text into $k$ segments during watermark embedding, with padding tokens removed, through the minimization of a loss function.

    \paragraph{Input and Output:}
    The input is a token string $\mathcal{S} = \{s^{(0)}, \ldots, s^{(L)}\}$ that carries the binary message $\mathcal{M}$, a predefined Bit Error Rate (BER) $\alpha$, and the number of segments $k$. 
    The output is the segmentation of the token string, represented as a list of segment boundaries.

    \paragraph{Preprocessing (Lines 1-3):}
    \begin{itemize}
        \item Compute the membership of each token $s^{(t)}$ in either $V_0$ or $V_1$, and store the results as binary values in the array `bit` (Line 3).
        \item Compute $P^{(t)}_0$, the probabilities associated with $s^{(t)}$, and store these values in the array $P$ (Line 4).
        \item Initialize two 2D arrays: $\mathcal{L}$ and $\text{prev}$ of size $k \times (L+1)$ to store the cumulative loss for each segmentation scenario and prefix information
    \end{itemize}

    \paragraph{Initialization of Loss Table (Lines 4-8):}Set all entries in $\mathcal{L}$ to infinity, ensuring that only valid segmentation paths are selected during subsequent calculations (Lines 6-8).

    \paragraph{Dynamic Programming for Segmentation (Lines 9-19):}
    \begin{itemize}
        \item For each segment index $t$ and each token position $i$, evaluate all possible preceding segment boundaries $j$ (Lines 10-12).
        \item Compute $\text{cost}[a][b]$, which represent the loss for the segment $(a, b)$ (Lines 13-14).
        \item If the total loss for the current segmentation path, $\mathcal{L}[t-1][a] + \text{cost}[a][b]$, is smaller than the existing loss at $\mathcal{L}[t][b]$, update $\mathcal{L}[t][b]$ and record the prefix position $j$ in $\text{prev}[t][b]$ (Lines 15-18).
    \end{itemize}

    \paragraph{Backtracking to Extract Segments (Lines 20-26):}
    \begin{itemize}
        \item Start from the last position of the token string and iteratively trace back using the boundary indices stored in $\text{prev}[t][b]$ to recover all segment boundaries (Lines 23-26).
        \item Append each segment as a pair of start and end indices $(\text{start}, \text{current})$ to the `Segments` list (Line 25).
    \end{itemize}

    \paragraph{Termination (Line 27):}
    Return the final list of segment boundaries, `Segments`, which represents the robust segmentation of the token string based on the minimum cumulative loss.

\section{DERMARK in Practical Scenarios}

\subsection{Text Length}

DERMARK is capable of efficiently embedding multi-bit watermarks into text generated by large language models (LLMs). However, in practical use, the text length and watermark capacity may not ideally align with the embedding requirements of DERMARK. Therefore, in this section, we discuss two extreme cases: excessively short and excessively long text lengths.

% Our method is designed with practical deployment scenarios in mind, particularly focusing on how to embed multi-bit watermarks effectively under varying text lengths. Here we discuss two edge cases and how our approach handles them compared to existing baselines.

\textbf{Case 1: Text length is insufficient to embed the full watermark.}  
In practice, when the text is too short or exhibits low entropy, its capacity to carry watermark bits becomes severely limited. In such cases, it is still feasible to dynamically adjust the watermark strength $\delta$, allowing for a trade-off between semantic fidelity and successful watermark embedding. However, in the worst-case scenario, DERMARK may still be unable to embed the entire preset watermark.

Importantly, due to its ability to adapt dynamically to local entropy, DERMARK achieves strictly higher embedding efficiency compared to fixed-length baseline methods. Consequently, if DERMARK fails to embed the watermark, fixed-length segmentation approaches are even more likely to fail or result in greater distortion.

\textbf{Case 2: Text length is significantly longer than the watermark bit length.}  
Such cases are conceptually inconsistent with the multi-bit watermarking goal, which aims to encode as much identifying information as possible. Embedding only a small number of bits into a very long text underutilizes the available capacity and runs counter to the intent of multi-bit watermarking. Nevertheless, if such a scenario arises, our method can still embed the watermark with minimal semantic disruption, concentrating on embedding within short segments and leaving most of the text untouched.

By contrast, baseline methods such as Balance-Marking distribute the watermark across the entire text to improve robustness. While this approach may enhance resilience to editing, it inevitably increases the overall semantic distortion. In such cases, choosing between the two approaches depends on application priorities: our method favors semantic preservation and high information density, while baselines may be preferred when robustness is paramount and semantic drift is tolerable.

Overall, our method targets a core challenge of multi-bit watermarking: how to maximize watermark density and semantic fidelity in capacity-limited settings. In capacity-abundant scenarios, the trade-offs become more nuanced and context-dependent, with our method and baseline approaches each offering unique advantages.

\subsection{Variable-Length Watermark}
In the context of this study, we consider that watermarks are embedded into the outputs of different LLMs as user-specific identifiers. Consequently, we assume that multi-bit watermarks are of fixed length $K$. However, in practice, there may be cases where watermarks are of variable length, meaning that the goal is to embed as many watermark bits as possible into the LLM-generated text.

In this context, the advantages of DERMARK can be fully leveraged. The method enables real-time assessment of the generated text to determine whether there is sufficient capacity to embed watermark bits. This dynamic decision-making process ensures that each watermark bit is successfully embedded without significantly affecting the semantics of the text, while simultaneously maximizing the number of watermark bits embedded within the text.

\section{Experimental Instructions and Supplements}
To facilitate reproducibility and further research, we publicly release the complete implementation of DERMARK at the following anonymous repository:\begin{links}
    \link{Code}{https://anonymous.4open.science/r/DERMARK-2F45}
\end{links}

\subsection{Explanation and Selection of Hyperparameters}

$\alpha$. The parameter $\alpha$ represents the confidence threshold used in the embedding inequality (Eq. (5)), controlling how strictly a segment must satisfy the bit alignment condition. A smaller $\alpha$ leads to more conservative segment acceptance, requiring a higher proportion of aligned tokens. In our experiments, we systematically varied $\alpha$ within the range $[0.8, 0.95]$ to analyze its effect on watermark capacity. For all other experiments, we report results across this range to capture the full performance profile under different confidence levels.

$\delta$. The watermark strength $\delta$ determines the magnitude of the bias applied to the logits of green or red tokens during watermark embedding. A higher $\delta$ increases the separation between token distributions, thus improving embedding robustness and detection reliability. In our capacity and quality experiments, we sweep over a wide range of $\delta$ values to assess its impact. For all other evaluations, we follow the default setting in KGW and fix $\delta = 1$, which provides a good balance between imperceptibility and robustness.

Our watermark extraction relies on the inequality condition in Eq. (5) to identify valid watermark segments. However, since the generated tokens are discrete and sampled probabilistically, the observed token proportions in each segment cannot strictly satisfy the inequality. Even in clean, watermarked texts, a small positive bias is inevitable, as the sampling process introduces stochastic deviation. Furthermore, if the text has undergone insertion or deletion attacks, additional perturbations will lead to unknown deviations in token distributions. To account for these effects, we introduce a bias-correction term $\epsilon_s$ in the extraction loss (cf.~Eq. (7)), which tolerates small violations of the inequality caused by stochasticity or editing.

In addition to inequality violations, another source of systematic deviation arises from the fact that each token only has a $(1 - \alpha)$ probability of being sampled from the intended list (e.g., the green list for embedding bit-1). As a result, even under ideal segmentation, each segment may naturally contain an $\alpha$-fraction of mismatched (e.g., red list) tokens. To accommodate this, we introduce a second correction term $\epsilon_c$(cf.~Eq. (8)), which adjusts for the expected level of token-type noise under the bit embedding distribution.

Rather than setting $\epsilon_s$ and $\epsilon_c$ manually, we adopt an iterative estimation strategy for both. In each iteration, we first compute the optimal segmentation via dynamic programming. Then, $\epsilon_s$ is updated based on the average residual in Eq. (5) across all segments, while $\epsilon_c$ is updated based on the observed deviation between empirical and expected token-type distributions. This process is repeated until both parameters converge, ensuring robustness to both stochastic and adversarial perturbations.

We also include other hyperparameters to handle estimation stability and loss balancing. The smoothing parameter $\lambda$ is used to stabilize the estimation of token-type priors ($P_{k1}$ and $P_{k0}$) in Eq. (6), particularly in short segments. This parameter is selected through manual tuning on validation data to balance numerical stability and estimation accuracy. 

Finally, $\beta$ is a tunable coefficient that controls the relative weight of the two components in Eq. (9), balancing segment-wise fidelity and global bit alignment. After extensive experimentation, we found that setting $\beta$ to 34 yields the best overall performance.

\subsection{Performance on high-entropy datasets}
% 这些数据往往有冗余，就是100词嵌入4个bit的情况
To strengthen the rigor of our evaluation, we construct a subset consisting of the top 25\% of datasets with the highest watermark detection rates under the baseline, which we refer to as the "high-entropy dataset." As illustrated in Fig.~\ref{fig:top}, while DERMARK on \texttt{OPT-1.3b} exhibits comparable performance to Balance-Marking, DERMARK on \texttt{LLaMA-2-7b} significantly outperforms Balance-Marking. This demonstrates that, as the model size increases, even for texts with high watermark capacity, Balance-Marking is more significantly impacted, whereas DERMARK remains largely unaffected. This observation strongly underscores the practical advantages of DERMARK.
In fact, the high-entropy dataset corresponds precisely to the case discussed earlier, where the generated text possesses sufficient capacity to embed multiple watermark bits. Most samples in this subset offer ample entropy, allowing reliable and effective multi-bit watermark embedding.

\begin{figure}[t]
 \centering
    \includegraphics[width=\linewidth]{ 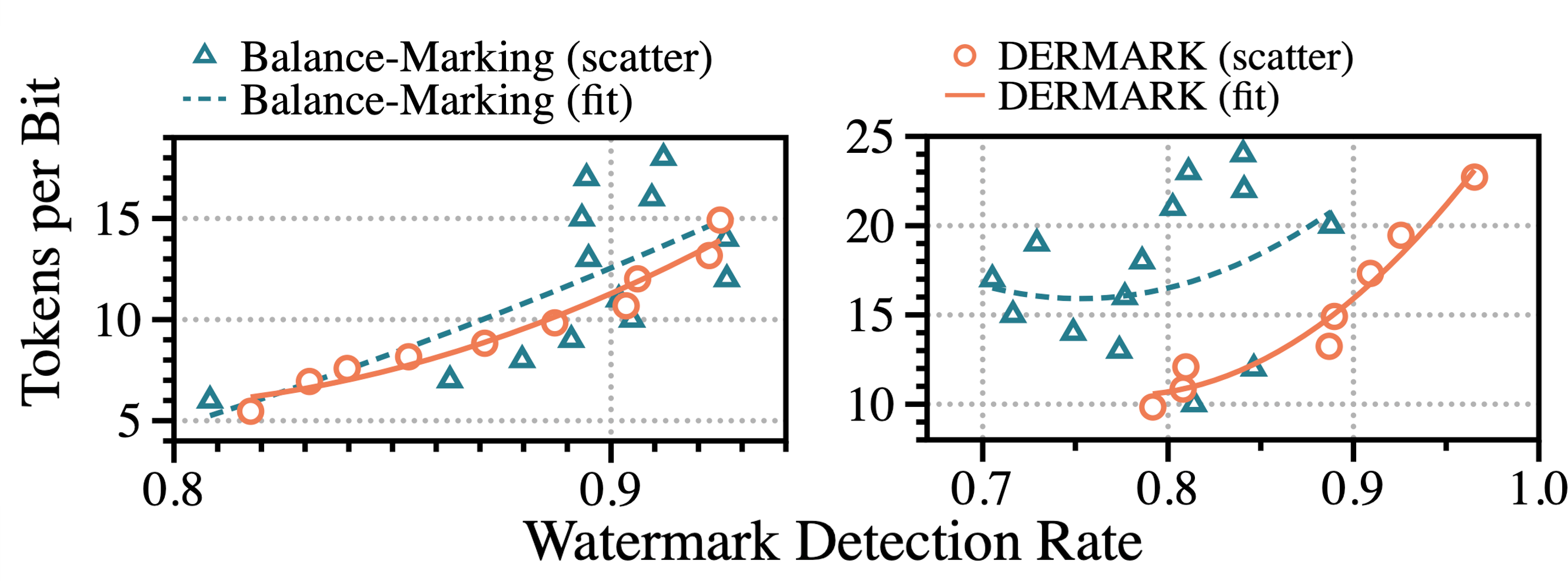}
    \caption{Capacity Comparison across high-entropy dataset on \texttt{OPT-1.3b} (left) and \texttt{LLaMA-2-7b} (right).}
    \label{fig:top}
\end{figure}

\subsection{Large-scale Model Deployment}

\begin{figure}[t]
    \centering
    \begin{subfigure}[b]{0.46\linewidth}
        \includegraphics[width=\linewidth]{ 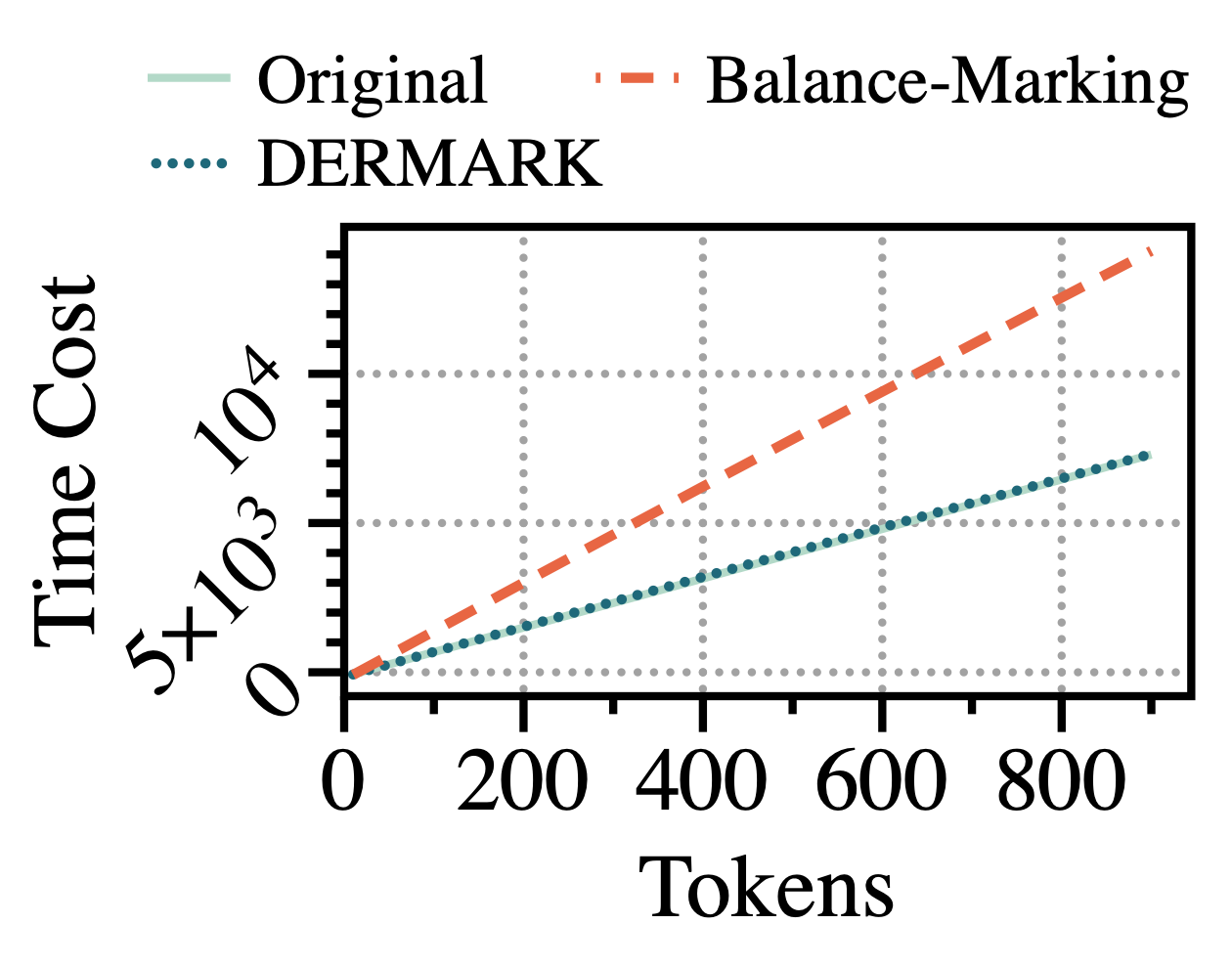}
        \caption{Text Generation}
        \label{fig:tim_generation_70b}
    \end{subfigure}
    \hfill
    \begin{subfigure}[b]{0.46\linewidth}
        \includegraphics[width=\linewidth]{ 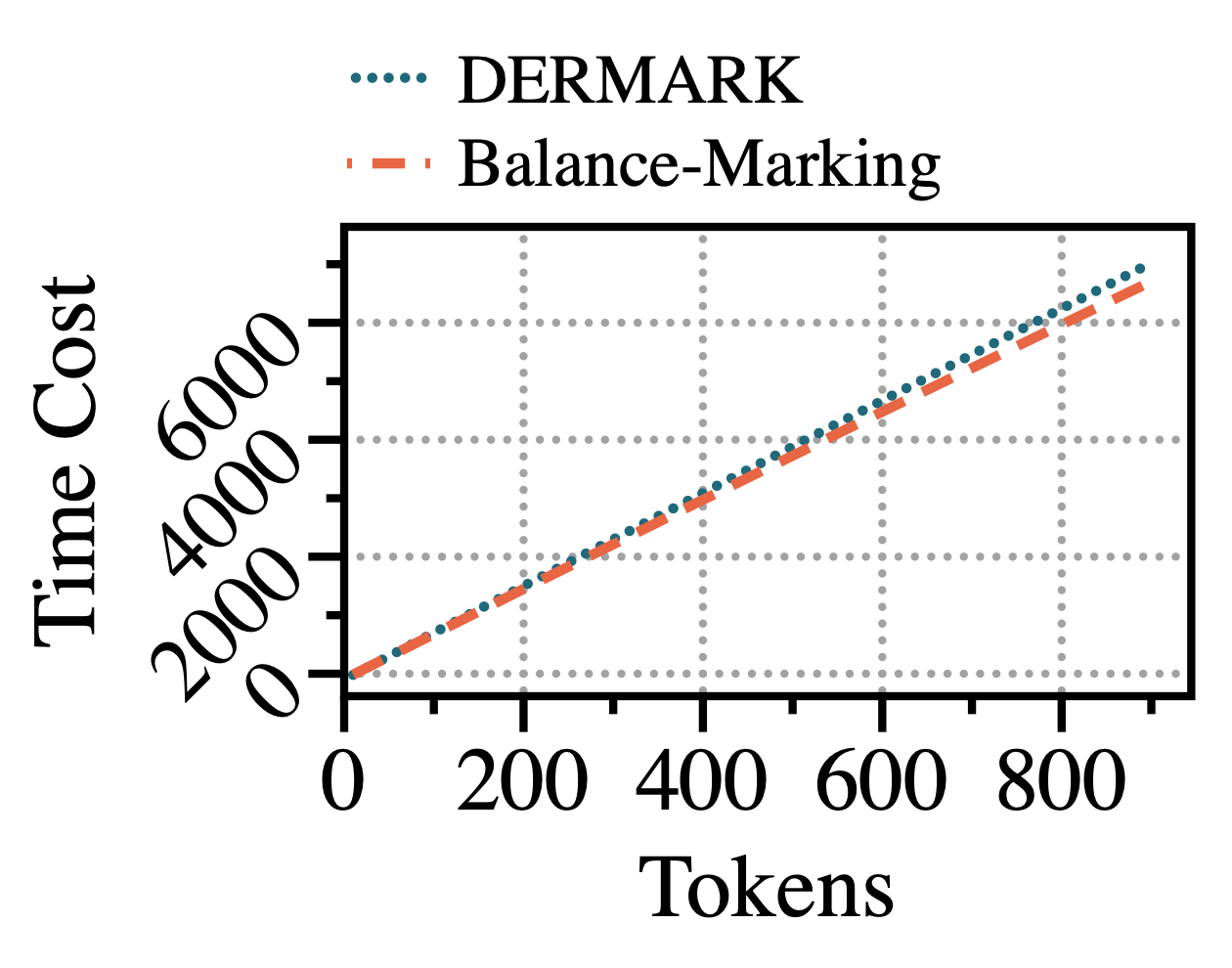}
        \caption{Watermark Extraction}
        \label{fig:tim_extraction_70b}
    \end{subfigure}
    \caption{Comparison of time cost in text generation and watermark extraction for \texttt{Llama-2-70B}.}
    \label{fig:time_cost_70b}
\end{figure}

To further demonstrate the practicality of DERMARK on large-scale LLMs, we evaluated the additional overhead on \texttt{LLaMA-2-70b} \cite{touvron2023llama}. The experimental results, as illustrated in Fig. \ref{fig:time_cost_70b}, show that the extra time overhead introduced by DERMARK during the watermark embedding phase is negligible relative to the total inference cost. In contrast, Balance-Marking incurs nearly double the overhead. While this issue may not be as pronounced for smaller LLMs where reasoning is faster, the impact becomes significantly more detrimental as the model scale increases and inference time lengthens. For the watermark extraction phase, the overhead of DERMARK is nearly identical to that of Balance-Marking.

In summary, the additional overhead introduced by DERMARK is negligible across LLMs of varying scales, underscoring its superior practicality.

\section*{Ethical Considerations}
In developing DERMARK to advance the embedding of multi-bit watermarking into LLMs, we have taken great care in our data practices. All datasets and models used in this study were derived from well-recognized previously published works, ensuring that they do not contain personally identifiable information. In addition, the evaluation benchmarks we employ are consistent with those established in previous studies, effectively eliminating the risk of privacy violations or data breaches.

\end{document}